\documentclass[conference,letterpaper]{IEEEtran}

\usepackage{amsfonts}
\usepackage{indentfirst}
\usepackage{float}
\usepackage{amsmath,amsthm}
\usepackage{mathtools}

\usepackage{amssymb}
\usepackage{color}
\usepackage{graphicx}
\usepackage[para]{threeparttable}
\newtheorem{theorem}{Theorem}

\newtheorem{lemma}{Lemma}
\newtheorem{Example}{Example}

\usepackage{threeparttable}
\usepackage{multirow}
\theoremstyle{definition}

\usepackage{booktabs, array, multirow, caption, makecell, tabularx}

\usepackage{cite}
\usepackage{dsfont}
\usepackage{stackengine}
\def\delequal{\mathrel{\ensurestackMath{\stackon[1pt]{=}{\scriptstyle\Delta}}}}

\begin{document}

\title{A New Family of Perfect Polyphase Sequences with Low Cross-Correlation}
%\thanks{}



\author{   ~\IEEEmembership{}
\thanks{}
\thanks{}% <-this % stops a space

}\author{%
  \IEEEauthorblockN{Dan Zhang}
  \IEEEauthorblockA{%Dept. of Information Security and Communication Technology\\
  Norwegian university of science and technology\\
                    %Gjøvik 2815, Norway\\
                    Email: dan.zhang@ntnu.no} 
                      \and
\IEEEauthorblockN{Staal Amund Vinterbo}
  \IEEEauthorblockA{Norwegian university of science and technology\\
                    %Gjøvik 2815, Norway\\
                    Email: staal.vinterbo@ntnu.no}
}

\date{}
\maketitle

\begin{abstract}

Spread spectrum multiple access systems demand minimum possible cross-correlation between the sequences within a set of sequences having good auto-correlation properties. Through a connection between generalised Frank sequences and Florentine arrays, we present a family of perfect sequences with low cross-correlation having a larger family size, compared with previous works. In particular,  the family size can be equal to the square root of the period when the period of the perfect sequences is even. In contrast, the number of the perfect sequences of even period  with low cross-correlation is equal to one in all previous works.

 %In this paper, we present a family of perfect sequences with low cross-correlation based on generalised Frank sequences and Florentine arrays. We derive a larger family size compared with previous work.  When the period of the perfect sequences is even, the family size can be equal to the square root of the period. In contrast, the number of the perfect sequences of even period  with low cross-correlation is equal to one in all previous works.

\end{abstract}


%%%%%%%%%%%%%%%%%%%%%%%%%%%%%%%%%%%%%%%%%%%%%%%%%%%%%%%

\section{Introduction}

Sequences and their properties have been widely studied in  different research areas because many applications depend on their characteristics. %Correlation functions are a measure of similarity among sequences. 
Sequences with desirable correlation properties have been used in communication systems and radar systems for identification, synchronization, ranging, and interference mitigation \cite{Golomb2005SignalDF}.  
 In Code-Division Multiple-Access systems, low cross-correlation between the desired and interfering users is important to suppress multi-user interference. Good auto-correlation properties are important for reliable initial synchronization and separation of the multi-path components. Moreover, the number of available sequences should be sufficiently large so that it can accommodate enough users. Therefore, it is of great interest to design families of sequences with large family size and low correlation.

The periodic {\it cross-correlation} value of two complex sequences ${\bf u}=\{u(t)\}_{t=0}^{N-1}$ and  ${\bf v}=\{v(t)\}_{t=0}^{N-1}$ of period $N$ at
shift $\tau$ is defined as
\begin{eqnarray*}
R_{{\bf u}, {\bf v}}(\tau)=\sum_{t=0}^{N-1}u(t+\tau) v^*(t), ~~~~0\leq \tau
<N,
\end{eqnarray*}
where $N$ is a positive integer, $t+\tau$  is taken modulo $N$, and $v^*(t)$ is the
complex conjugate of  the complex number $v(t)$. When two sequences ${\bf u}$ and ${\bf v}$ are identical, the  periodic  cross-correlation function  is called {\it auto-correlation}	function, and is denoted by $R_{{\bf u}}(\tau)$. A sequence is said to be \textit{perfect}
if all the out-of-phase periodic auto-correlation
coefficients are zero, i.e., $R_{{\bf u}}(\tau)=0$ for  $\tau \not \equiv 0  \bmod N$.

Let  $\mathcal{S}$ be a set  of $M$ sequences of period $N$.  The maximum out-of-phase periodic auto-correlation magnitude is denoted by $R_{a}$ and defined by
$
R_{a} =  \max\{ \vert R_{{\bf s}_{i}}(\tau) \vert : {\bf s}_{i} \in \mathcal{S}, 0 < \tau < N \}.
$
The  maximum periodic cross-correlation magnitude is denoted by $R_{c}$ and defined by
$
R_{c} =  \max\{ \vert R_{{\bf s}_{i}, {\bf s}_{j}}(\tau) \vert : {\bf s}_{i} \neq {\bf s}_{j} \in \mathcal{S},  0 \leq \tau < N \}.
$  A lower bound on $R_{\max}= \max (R_{a}, R_{c})$  given by Welch \cite{Welch} is
$
R_{\max}\geq N\sqrt{\frac{M-1}{MN-1}}. 
$
Due to the above bound, it is of great interest to design  a sequence set with $ \sqrt{N}\leq R_{\max}\leq c\sqrt{N}$, where $c $ is a small constant and $N$ is the period of the sequences in the family. We call such a set \textit{a family of sequences with low correlation}.  Excellent surveys and fundamental discussions on this topic can be found~\cite{ HellesethKumar, Garg2009}.
 
We are particularly interested in families of perfect sequences with low correlation. Perfect sequences have ideal auto-correlation, i.e., $R_{a} =0$ in these families. Another bound called the Sarwate bound \cite{Sarwate} implies that $R_{c} \geq \sqrt{N}$. A set of perfect sequences meeting this bound is called \textit{an optimal set of perfect sequences}. Extensive research has been done on how to generate optimal families of perfect sequences \cite{Sarwate, Alltop, Alltop1, Popovic2, Suehiro, chirplike, Gabidulin1993, MSong, Mow, KPark, DanISIT2020,DanTIT22,SSong22}. In these works \cite{Sarwate, Alltop, Alltop1, Popovic2, Suehiro, chirplike, Gabidulin1993, MSong, Mow, KPark}, the number of perfect sequences with optimal cross-correlation is equal to $p-1$, where $p$ is the smallest prime divisor of  the period $N$.  Recent works \cite{DanISIT2020,DanTIT22,SSong22} show that the family size can be larger than $p-1$, and is determined by the existence of well-studied combinatorial objects, circular Florentine arrays. However, these constructions based on circular Florentine arrays that produce the desired large families can only do so for odd periods. When the period is even, the constructions yield families of size one.

%all these constructions are trivial when the period of the sequence is  even. No pair of perfect sequences of even period with optimal cross-correlation has been reported, which means the family size in all the previous works is equal to one in this case. 

In this paper, we propose a construction of perfect sequences with low correlation based on non-circular Florentine arrays. This construction allows us to derive a family of perfect sequences with $R_{c}=2\sqrt{N}$, where $N$ is the period of the sequences. 
The family size depends on the existence of Florentine arrays, which is greater than that in the previous works. In particular, the number of perfect sequences with low cross-correlation can be $\sqrt{N}$ for even $N$. Table~\ref{Listofwork} relates the above previous works to our results.

\section{Preliminaries}\label{sec-pre}

\subsection{Florentine arrays}
An $m \times n$  \textit{(circular) Tuscan-$k$  array }has $m$ rows and $n$ columns such that 1) each row is a permutation of
$n$ symbols and 2) for any two symbols $a$ and $b$, and for
each $t$ from $1$ to $k$, there is at most one row in which $b$ occurs  $t$ steps (circularly) to the right of $a$. In  particular, a (circular) Tuscan-$(n-1)$ array is referred to as a (circular) Florentine array. When $m=n$, we call them (circular) Tuscan squares and (circular) Florentine squares, respectively. 

%A latin square of order $n$ is an $n\times n$ array with $n$ distinct symbols having the property that each symbol occurs exactly once in each row and column. A vatican square is a  Florentine square that is also a latin.
 
For each positive integer $n \geq2$, 
 we denote $F(n)$ the maximum number such that an $F(n) \times n$ Florentine array exists and $F_{c}(n)$ the maximum number such that an $F_{c}(n) \times n$ circular Florentine array exists. By definition, $F(n) \geq F_{c}(n)$ for all $n$, because any circular Florentine arrays are also Florentine arrays.
 
\begin{lemma} \cite{HSong}  \label{CircularFB}
\begin{itemize}
\item[(1)] $F_{c}(n)=1$  when $n$ is even, and
\item[(2)] $p-1 \leq F_{c}(n) \leq n-1$, where $p$ is the smallest prime factor of $n$, and 
\item[(3)] $F_{c}(n) = n-1$  when $n$ is a prime.
\end{itemize}
\end{lemma}

\begin{lemma}\cite{TAYLOR1991}\label{FB}
\begin{itemize}
\item[(1)]$F(n)\leq n$, and
\item[(2)]$F(n) \geq F_{c}(n+1)$ for all $n$, and
\item[(3)] $F(n) \geq n-1$ and $F(n-1) = n-1$ when $n$ is a prime, and 
\item[(4)] $F(n) \geq p-1$ and $F(n-1) \geq p-1$, where $p$ is the smallest prime divisor of $n$.
\end{itemize}
\end{lemma}

Note that  the fact that $F(n) \geq F_{c}(n+1)$ for all $n$, is because any $F_{c}(n+1)$ circular Florentine rows on $n+1$ symbols can lead to the same number of rows on $n$ symbols by deleting any one symbol in each row.  With this fact and the lower bound on  $F_{c}(n)$, one can derive both $F(n) \geq p-1$ and $F(n-1) \geq p-1$, where $p$ is the smallest prime divisor of 
$n$. It follows that $F(n) \geq n-1$ and $F(n-1) = n-1$ when $n$ is a prime.

\begin{table}[!h]
\caption{ $6 \times 6$  and $6 \times 7$ Florentine arrays \cite{Colbourn}}

\parbox{.4\linewidth}{
\centering
\begin{tabular}{|c|c|c|c|c|c|}
\hline
$ 0$ & $2$ & $1$ & $4$ & $5$ & $3$  \\ \hline
$ 1$ & $3$ & $2$ & $5$ & $0$ & $4$  \\ \hline
$ 2$ & $4$ & $3$ & $0$ & $1$ & $5$   \\ \hline
$ 3$ & $5$ & $4$ & $1$ & $2$ & $0$   \\ \hline
$ 4$ & $0$ & $5$ & $2$ & $3$ & $1$   \\ \hline
$5$ & $1$ & $0$ & $3$ & $4$ & $2$   \\ \hline
\end{tabular}
\label{Florentineexample}
}	
\hfill
\parbox{.5\linewidth}{
\centering
\begin{tabular}{|c|c|c|c|c|c|c|}
\hline
$ 1$ & $2$ & $3$ & $4$ & $5$ & $6$ & $0$  \\ \hline
$ 2$ & $5$ & $0$ & $4$ & $3$ & $1$ & $6$  \\ \hline
$ 3$ & $0$ & $1$ & $4$ & $6$ & $5$ & $2$  \\ \hline
$ 5$ & $3$ & $6$ & $4$ & $0$ & $2$ & $1$  \\ \hline
$ 6$ & $1$ & $5$ & $4$ & $2$ & $0$ & $3$  \\ \hline
$0$ & $6$ & $2$ & $4$ & $1$ & $3$ & $5$  \\ \hline
\end{tabular}	
}
\end{table}

To achieve the upper bound on $F(n)$,  it will be interesting to know when a Florentine square exists. The only known Florentine squares of order $n$ are Vatican squares and come form the so-called prime construction which essentially is from the multiplication 
table $\bmod~(n+1)$, where $n+1$ is prime. 
Exhaustive search for Florentine arrays has been done by many researchers. 
Taylor \cite{TAYLOR1991} gave a table of all possible values of $F(n)$ for $1 \leq n \leq 32$, which was later updated by Hong Yeop Song \cite{HYSong} (See Table \ref{FAtable}). For more works on Tuscan arrays, see~\cite{GolombT, Colbourn}.

\begin{table}[!h] 
	\centering
	\caption{Possible values of  $F(n)$  \cite{Colbourn}}
	\begin{tabular}{| c | c | c | c | c | c |}
		\hline $n$  & $F(n)$ & $n$  & $F(n)$ & $n$ & $F(n)$  \\
		\hline $1$  & $1$ & $11$  & $10$ & $21$ & $7, \cdots, 21$  \\
		\hline $2$  & $2$ & $12$ & $12$ & $22$ & $22$  \\
		\hline $3$  & $2$ & $13$  & $12,13$ & $23$ & $22,23$  \\
		\hline $4$  & $4$ & $14$  & $7, \cdots, 14$ & $24$ & $6, \cdots, 24$  \\
		\hline $5$  & $4$ & $15$  & $7, \cdots, 15$ & $25$ & $6, \cdots, 25$  \\
		\hline $6$  & $6$ & $16$  & $16$ & $26$ & $6, \cdots, 26$  \\
		\hline $7$  & $6$ & $17$  & $16,17$ & $27$ & $6, \cdots, 27$  \\
		\hline $8$  & $7$ & $18$  & $18$ & $28$ & $28$  \\
		\hline $9$  & $8$ & $19$  & $18, 19$ & $29$ & $28, 29$  \\
		\hline $10$  & $10$ & $20$  & $6, \cdots, 20$ & $30$ & $30$  \\		
		\hline
	\end{tabular}
	\label{FAtable}
\end{table}

Let $C$ be an $m \times n$  Florentine array on $\mathbb{Z}_{n}$, where  $\mathbb{Z}_{n}$ denotes the ring of integers modulo $n$. The rows are indexed as $1$ to~$m$.  By definition,  each row is a permutation over $\mathbb{Z}_{n}$, denoted by $\beta_i$ for $1 \leq i \leq m$. These permutations have the following property.

\begin{lemma} \label{NumbofSolequation}
For $1 \leq {i}, {j} \leq m$ such that $i \neq j$ and $l \in \mathbb{Z}_{n}$, let
\begin{equation*} 
\mathcal{N}^l_{(i, j)}=\{t \in \mathbb{Z}_{n} ~| ~\beta_{i}(t) = \beta_{j}((t+l) \bmod n) \}.
\end{equation*} 
Then $|\mathcal{N}^l_{(i, j)}|\leq 2$ and the bound is tight.
\end{lemma}

\begin{proof} 
Let addition be in $\mathbb{Z}$ and let $\delta(x)=\mathds{1}(x\geq n)$ where $\mathds{1}$ is the indicator function. Then $\delta$ indicates whether argument $x$
 ``wraps around" modulo~$n$. 
 
 For any $l \in \mathbb{Z}_{n}$ and ${i} \neq {j}$, let
 $ t, t' \in \mathcal{N}^l_{(i, j)} $ and $ t \neq t'$. First we prove that $\delta(t + l) \neq \delta(t' + l)$.
Without loss of generality,  let $0 \leq t < t' < n$.
Since $t, t' \in \mathcal{N}^l_{(i, j)} $, we have 
\begin{align*}
     \beta_{i}(t) &= \beta_{j}((t + l) \bmod n), \; \text{and} \\
     \beta_{i}(t') &= \beta_{j}((t' + l) \bmod n).       
\end{align*}
We assume that  $\delta(t +l) = \delta(t' +l ) = c $. It follows that
\begin{align*}
&((t' + l) \bmod n) - ((t + l) \bmod n) \\ & = (t' + l - cn) - (t + l - cn) \\ &=  t' - t. 
\end{align*}
Then the pair $(\beta_{i}(t), \beta_{i}(t') )= (\beta_{j}((t + l) \bmod n), \beta_{j}((t' + l) \bmod n)) \delequal(a, b) $ with $b$
being the $(t'-t)$-th step to the right of $a$  appear at two different rows $i$ and $j$, which contradicts the definition of  Florentine arrays. Therefore, $\delta(t + l) \neq \delta(t' + l)$ for   $ t, t' \in \mathcal{N}^l_{(i, j)}. $

Now we show that $|\mathcal{N}^l_{(i, j)}|\leq 2$.
%for $ l \in \mathbb{Z}_{n}$ and $ {i} \neq {j}$. 
Assume on the contrary, 
%$|\mathcal{N}^l_{(i, j)}|\geq 3$ for $i \neq j$. 
%Without loss of generality, let  
there exist $t , t',  t'' \in \mathcal{N}^l_{(i, j)}$  with
 $0 \leq t < t' < t'' < n$.  Since $\delta$ is a two-valued function, at least two of the elements $\delta(t+l)$, $\delta(t' +l)$ and $\delta(t'' +l)$ must share the same value. This contradicts the fact that  $\delta(t + l) $ and $ \delta(t' + l)$ can not be the same for  any $ t, t' \in \mathcal{N}^l_{(i, j)}.$
Consequently, we have  $|\mathcal{N}^l_{(i, j)}|\leq 2$ for $ l \in \mathbb{Z}_{n}$ and $ {i} \neq {j}$. 

For the $6 \times 6$ Florentine array in Table \ref{Florentineexample}, $\mathcal{N}^2_{(1, 2)}=\{3, 5\}$,  demonstrating that the bound is tight. \qedhere
\end{proof}

\subsection{Perfect polyphase sequences}
A \textit{polyphase sequence} is a sequence whose elements are  all complex roots of unity of the form $exp(i2\pi x)$ where $x$ is a rational number and $i = \sqrt{-1}$.  Many studies have been done on the constructions of perfect polyphase sequences.  Mow \cite{Mow} classified all known perfect polyphase sequences into four classes:  generalised Frank sequences \cite{Kumar},  generalised chirp-like sequences \cite{chirp-like},  Milewski sequences \cite{Milewski}, and  perfect polyphase sequences associated with generalised bent functions \cite{Chung}. Mow also proposed a unified construction of perfect polyphase sequences and conjectured that the unified construction describes all the perfect polyphase sequences that exist.

Generalized Frank sequences are a class of perfect polyphase sequences which are from one-dimensional bent function and were proposed by Kumar, Scholtz and Welch~\cite{Kumar}. These sequences  were first discovered by Frank and Zadoff~\cite{Frank} in the case $\sigma = 0$ and $\pi$ being the identity permutation.  Heimiller \cite{Heimiller}  found the sequences
$
\omega_{N^2}^{N\cdot \pi(t_{1}) ( t_{2} + h(t_{1}))}
$
for the case of prime $N$, where $h$ is also an arbitrary function on $\mathbb{Z}_{N}$. Generalized Frank sequences are a more general family, and are defined as follows.

\begin{lemma} \cite{Kumar}\label{bent function}
Let $N$ be a positive integer and $\omega_N$ be a primitive  $N$-th root of unity. Let 
\begin{itemize}
\item[(i)] $\pi$ be a permutation of elements in  $\mathbb{Z}_{N}$ and let
\item[(ii)] $\sigma$ be an arbitrary function from  $\mathbb{Z}_{N}$ to  $\mathbb{Z}_{N^2}$.
\end{itemize}
Then $
s(t)= \omega_{N^2}^{N\cdot t_{2} \pi(t_{1}) + \sigma(t_{1})}
$
where $t=t_{1}+ N \cdot t_{2} $, $0 \leq t_{1}, t_{2} < N$, is a perfect sequence of period $N^2$.
\end{lemma}

By Lemma \ref{bent function}, there are in total $N!N^{2m}$ perfect sequences of period $N^2$. In order to generate an optimal set from these sequences,  the maximum cross-correlation magnitude of any two distinct sequences  should be  $N$. There exist many studies on perfect sequences with optimal cross-correlation (see Table~\ref{Listofwork}). However, these constructions are trivial when~$N$ is even, which means no pair of perfect sequences of even period with optimal cross-correlation has been reported. In next section, we present a family of perfect sequences of period~$N^2$ based on Lemma \ref{bent function}, whose maximum cross-correlation magnitude of any two distinct sequences is $2N$.  The number of sequences in this family can be $N$ when $N$ is even.

\section{Families of perfect sequences with low cross-correlation}

\begin{table*}[t]  \label{table1}
\centering
 \begin{threeparttable}
\caption{Families of perfect polyphase sequences with low cross-correlation}
\small{
 \setlength\tabcolsep{2pt}
\begin{tabular}{|c|c|c|c|c|c|c|c|c|c|c|c|c|}

\hline References & \cite{Mow} & \cite{DanTIT22} \cite{SSong22} & \cite{Sarwate}  \cite{Alltop} & \cite{Popovic2}& \cite{Gabidulin1993}& \cite{Alltop1} & \cite{Suehiro} &\cite{Gabidulin1993} & \cite{KPark} & \cite{MSong} & \cite{DanTIT22} \cite{SSong22} & this paper  \\

\hline   \makecell{Class of\\perfect sequences} &  \multicolumn{2}{c|}{Unified construction} &  \multicolumn{3}{c|}{\makecell{Generalised chirp-like\\ polyphase sequences}}   & \multicolumn{7}{c|}{Generalised Frank sequences}    \\
\hline   \makecell{Period of \\perfect sequences} & $rm^2$ & $rm^2$ $(r\neq 1)$ & $N$  & $rm^2$ & $P^{2h+1}$ &$Q^2$  & $P^2$ &$P^{2h}$  & $P^2$ & $N^2$  & $N^2$ & $N^2$\\
\hline \makecell{The family \\ size} & $p-1$ &$\min\{r^{*}-1, F_{c}(m)\}$ & $p-1$ & $p-1$  & $p-1$& $\frac{p-1}{2}$ & $p-1$ &$p-1$ & $p-1$ &  $p-1$ & $F_{c}(N)$& $F(N)$  \\

\hline
\end{tabular}}
\label{Listofwork}
 \begin{tablenotes}
   \item $N$,  $r$, $m$ and $h$ are positive integers;
 \item $P$ is an  odd prime;
     \item $Q$ is an  odd integer;
   \item $p$ is the smallest prime divisor of the period;
   \item $r^{*}$ is the smallest prime divisor of $r$;
   \item $F_{c}(m)$ is the maximum number such that an $F_{c}(m) \times m$ circular Florentine array exists.
      \item $F(N)$ is the maximum number such that an $F(N) \times N$  Florentine array exists;
    \end{tablenotes}
\end{threeparttable}
\end{table*}

In this section, we build a connection between generalised Frank sequences and  Florentine arrays, which allows us to generate a family of perfect sequences with  a large family size and low cross-correlation.

Let $N$ be a positive integer. Let $C$ be an $F(N) \times N$  Florentine array over $ \mathbb{Z}_{N}$,  where 
$F(N)$ is the maximum number such that an $F(N) \times N$ circular Florentine array exists.  Let $\mathcal{A}=\{\beta_{1}, \beta_{2},  \cdots \beta_{F(N)}\}$ be a set of permutations over $ \mathbb{Z}_{N}$ from the rows of $C$.   A set of sequences of period $N^2$ is defined as
\begin{equation} \label{Frankseq}
\mathcal{S} = \{  {\bf s}_{i} ~\vert~ {\bf s}_{i}(t) = \omega_{N^2}^{N \cdot \beta_i(t_{1}) t_{2} + \sigma(t_{1})}, \beta_i \in \mathcal{A} \},
\end{equation}
 where $t=t_{1}+t_{2}\cdot N$, $0 \leq t_{1}, t_{2}  < N$, and $\sigma$ is an arbitrary function from $\mathbb{Z}_{N}$ to  $\mathbb{Z}_{N^2}$.

\begin{theorem} \label{maintheorem} 
The  set $\mathcal{S}$  defined by \eqref{Frankseq} is a family of perfect sequences of size $F(N)$  with $R_{c} = 2N$.
\end{theorem}

\begin{proof}
Since each $ \beta_i \in \mathcal{A} $ is a permutation over $\mathbb{Z}_{N}$, each sequence in $\mathcal{S}$ is perfect by Lemma \ref{bent function}. 
For any shift $0\leq \tau< N^2$, we rewrite $\tau=\tau_1+\tau_2\cdot N$, where $0 \leq \tau_1, \tau_2 < N$, and
define
\begin{eqnarray*}\label{eqn_delta}
\delta_{t_{1}, \tau_1}=\left\{ \begin{array}{ll}
0& \mbox{~if~} t_1+\tau_1<N, \\
1&\mbox{~if~} t_1+\tau_1 \geq N.
\end{array}
\right.
\end{eqnarray*}

 Let ${\bf s}_{i}$ and ${\bf s}_{j}$ be two  sequences in $\mathcal{S}$, where  $1 \leq i \neq j \leq F(N)$. The cross-correlation between ${\bf s}_{i}$ and ${\bf s}_{j}$ is given by
\begin{equation*} \label{the-main-proof}
\begin{aligned}
R_{{\bf s}_{i},{\bf s}_{j}}(\tau) = & \sum_{t=0}^{N^2-1}s_{i}(t+\tau){s_{j}}^*(t)\\
 = &\sum\limits_{t_{2}=0}^{N-1} \sum\limits_{t_{1}=0}^{N-1}
               \omega_{N^2}^{N \cdot \beta_{i}(t_{1}+\tau_1)(t_{2}+\tau_2+\delta_{t_{1}, \tau_1}) + \sigma( t_{1}+\tau_1) } \\    
              & \cdot \omega_{N^2}^{- (N \cdot \beta_{j}(t_{1}) t_{2} + \sigma(t_{1}))}\\
                =  &\sum\limits_{t_{1}=0}^{N-1} \omega_{N^2}^{ N \cdot \beta_{i}(t_{1}+\tau_1)(\tau_2+\delta_{t_{1}, \tau_1})+ \sigma(t_{1}+\tau_1)-\sigma(t_{1})}   \\
  & \cdot \sum\limits_{t_{2}=0}^{N-1} \omega_{N}^{(\beta_{i}( t_{1}+\tau_1) -\beta_{j}(t_{1})) t_{2}}.
\end{aligned}
\end{equation*}
The inner sum of the last identity above is zero unless
\begin{eqnarray*}\label{pifunction}
\beta_{i}(t_{1}+\tau_1) \equiv \beta_{j}(t_{1}) ~ \bmod ~N.
\end{eqnarray*}
Since $\beta_{i}$ and $\beta_{j}$  are two rows from a Florentine array, the above equation has at most two solutions in
$\mathbb{Z}_{N}$  for   $\forall \tau_1 \in \mathbb{Z}_{N}$ and $ i \neq j$   by Lemma \ref{NumbofSolequation}. Therefore, we have $\vert R_{{\bf s}_{i},{\bf s}_{j}}(\tau)  \vert \leq 2N$ for all $0 \leq  \tau < N^2-1$ and  $ i \neq j $.
\end{proof}

 \begin{Example}
Let $N=6$ and  a $6 \times 6$  Florentine array is provided in Table \ref{Florentineexample}. Let  $\mathcal{A} =\{\beta_{1}, \beta_{2}, \cdots, \beta_{6}\}$ denote the set of permutations from the rows of the  Florentine array. For simplicity, let $\sigma =0$. Then a set of sequences of period $225$ is defined as
\begin{equation*} 
\mathcal{S} = \{  {\bf s}_{i} ~\vert~ {\bf s}_{i}(t) = \omega_{15}^{ \pi_i(t_{1}) t_{2} }, 1 \leq i \leq 6\},
\end{equation*}
 where $t=t_{1}+t_{2}\cdot 6$, $0 \leq t_{1}, t_{2}  < 6$, $\pi_i \in \mathcal{A}$ for $1 \leq i \leq 6$.
It is  verifiable that
\begin{itemize}
\item each sequence is a perfect sequence of period $36$;
and
\item $\vert R_{{\bf s}_{i},{\bf s}_{j}}(\tau) \vert \leq 12 $ for any  $~0\leq \tau< 6$,
 $1\leq i \neq j \leq 6$.
\end{itemize}
Therefore, the set $\mathcal{S}$ is a family of $6$ perfect sequences of period $36$ with $R_{c} = 12$, which are consistent with Theorem \ref{maintheorem}.
\end{Example}
  
Given an $F(N) \times N$  Florentine array, we can get a family of  $F(N)$  generalised Frank sequences of period $N^2$, where $N$ is a positive integer and $F(N)$ is the maximum number such that an $F(N) \times N$  Florentine array exists. Table \ref{Listofwork} gives a list of known results. Note that  $R_{c}$ in all the other works  is equal to the square root of the period, which means optimal cross-correlation. However, the family size in the previous works is either determined by the smallest prime divisor of the period or the existence of circular Florentine arrays. The properties of Florentine arrays in Lemma \ref{FB} implies  that the family size is larger in this paper. Furthermore, the number of rows  in a Florentine array for even  $N$, can be equal to $N$ (see Table \ref{FAtable}), which allows us to derive perfect sequences with low cross-correlation with family size~$N$.
In contrast,  the family size in all the other works is equal to one when the period of the sequences is even.

%The properties of Florentine arrays in Lemma \ref{FB} implies  that the family size is larger that in the previous works (see Table \ref{Listofwork}).

\section{Conclusion}

We derived a family of perfect sequences with low  cross-correlation based on Florentine arrays. The number of  the perfect sequences depends on the existence of  Florentine arrays. 
The properties of Florentine arrays assure that the family size is  larger than that in the previous works. 
The  previous constructions are trivial when the period of the perfect sequences is even. In this work, a small compromise on the optimality of the cross-correlation allows us to derive an non-trivial construction of perfect sequences with low cross-correlation for even period.

%%%%%%%%%%%%%%%%%%%%%%%%%%%%%%%%%%%%%%%%%%%%%%%%%%%%%%%
%\subsection*{Acknowledgements}
%The authors gratefully acknowledgement constructive criticisms by  anonymous reviewers and specially thank prof. Zhengchun Zhou for his   %helpful discussions and comments.
%%%%%%%%%%%%%%%%%%%%%%%%%%%%%%%%%%%%%%%%%%%%%%%%%%%%%%%

\section*{Acknowledgment}

This work was  supported in part by Innlandet Fylkeskommune.
%%%%%%%%%%%%%%%%%%%%%%%%%%%%%%%%%%%%%%%%%%%%%%%%%%%%%%%

\bibliographystyle{IEEEtran}
\bibliography{references}

\begin{thebibliography}{10}
\providecommand{\url}[1]{#1}
\csname url@samestyle\endcsname
\providecommand{\newblock}{\relax}
\providecommand{\bibinfo}[2]{#2}
\providecommand{\BIBentrySTDinterwordspacing}{\spaceskip=0pt\relax}
\providecommand{\BIBentryALTinterwordstretchfactor}{4}
\providecommand{\BIBentryALTinterwordspacing}{\spaceskip=\fontdimen2\font plus
\BIBentryALTinterwordstretchfactor\fontdimen3\font minus
  \fontdimen4\font\relax}
\providecommand{\BIBforeignlanguage}[2]{{%
\expandafter\ifx\csname l@#1\endcsname\relax
\typeout{** WARNING: IEEEtran.bst: No hyphenation pattern has been}%
\typeout{** loaded for the language `#1'. Using the pattern for}%
\typeout{** the default language instead.}%
\else
\language=\csname l@#1\endcsname
\fi
#2}}
\providecommand{\BIBdecl}{\relax}
\BIBdecl

\bibitem{Golomb2005SignalDF}
S.~Golomb and G.~Gong, \emph{Signal Design for Good Correlation: for Wireless
  Communication, Cryptography, and Radar}.\hskip 1em plus 0.5em minus
  0.4em\relax Cambridge University Press, 2005.

\bibitem{Welch}
L.~{Welch}, ``Lower bounds on the maximum cross correlation of signals
  (corresp.),'' \emph{IEEE Transactions on Information Theory}, vol.~20, no.~3,
  pp. 397--399, May 1974.

\bibitem{HellesethKumar}
T.~{Helleseth} and P.~V. {Kumar}, ``Sequences with low correlation,'' \emph{In
  V. S. Pless and W. C. Huffman, editors, Handbook of Coding Theory, Vol. I,
  II, chapter 21}, p. 1765â€“1853, 1998.

\bibitem{Garg2009}
G.~Garg, T.~Helleseth, and P.~V. Kumar, \emph{Recent Advances in
  Low-Correlation Sequences}.\hskip 1em plus 0.5em minus 0.4em\relax Boston,
  MA: Springer US, 2009, pp. 63--92.

\bibitem{Sarwate}
D.~{Sarwate}, ``Bounds on crosscorrelation and autocorrelation of sequences
  (corresp.),'' \emph{IEEE Transactions on Information Theory}, vol.~25, no.~6,
  pp. 720--724, November 1979.

\bibitem{Alltop}
W.~{Alltop}, ``Complex sequences with low periodic correlations (corresp.),''
  \emph{IEEE Transactions on Information Theory}, vol.~26, no.~3, pp. 350--354,
  May 1980.

\bibitem{Alltop1}
------, ``Decimations of the frank-heimiller sequences,'' \emph{IEEE
  Transactions on Communications}, vol.~32, no.~7, pp. 851--853, July 1984.

\bibitem{Popovic2}
B.~M. {Popovi\'{c}}, ``Generalized chirp-like polyphase sequences with optimum
  correlation properties,'' \emph{IEEE Transactions on Information Theory},
  vol.~38, no.~4, pp. 1406--1409, July 1992.

\bibitem{Suehiro}
N.~{Suehiro} and M.~{Hatori}, ``Modulatable orthogonal sequences and their
  application to ssma systems,'' \emph{IEEE Transactions on Information
  Theory}, vol.~34, no.~1, pp. 93--100, Jan 1988.

\bibitem{chirplike}
B.~M. {Popović}, ``Generalized chirp-like polyphase sequences with optimum
  correlation properties,'' \emph{IEEE Transactions on Information Theory},
  vol.~38, no.~4, pp. 1406--1409, July 1992.

\bibitem{Gabidulin1993}
E.~M. {Gabidulin}, ``Non-binary sequences with the perfect periodic
  auto-correlation and with optimal periodic cross-correlation,'' in
  \emph{Proceedings. IEEE International Symposium on Information Theory}, 1993,
  pp. 412--412.

\bibitem{MSong}
M.~K. {Song} and H.~{Song}, ``A construction of odd length generators for
  optimal families of perfect sequences,'' \emph{IEEE Transactions on
  Information Theory}, vol.~64, no.~4, pp. 2901--2909, April 2018.

\bibitem{Mow}
W.~H. {Mow}, ``A new unified construction of perfect root-of-unity sequences,''
  in \emph{Proceedings of ISSSTA'95 International Symposium on Spread Spectrum
  Techniques and Applications}, vol.~3, 1996, pp. 955--959.

\bibitem{KPark}
K.~{Park}, H.~{Song}, D.~S. {Kim}, and S.~W. {Golomb}, ``Optimal families of
  perfect polyphase sequences from the array structure of fermat-quotient
  sequences,'' \emph{IEEE Transactions on Information Theory}, vol.~62, no.~2,
  pp. 1076--1086, Feb 2016.

\bibitem{DanISIT2020}
D.~Zhang and T.~Helleseth, ``New optimal sets of perfect polyphase sequences
  based on circular florentine arrays,'' in \emph{2020 IEEE International
  Symposium on Information Theory (ISIT)}, 2020, pp. 2921--2925.

\bibitem{DanTIT22}
------, ``Sequences with good correlations based on circular florentine
  arrays,'' \emph{IEEE Transactions on Information Theory}, vol.~68, no.~5, pp.
  3381--3388, 2022.

\bibitem{SSong22}
M.~K. Song and H.-Y. Song, ``New framework for sequences with perfect
  autocorrelation and optimal crosscorrelation,'' \emph{IEEE Transactions on
  Information Theory}, vol.~67, no.~11, pp. 7490--7500, 2021.

\bibitem{HSong}
H.-Y. Song, ``The existence of circular florentine arrays,'' \emph{Computers \&
  Mathematics with Applications}, vol.~39, no.~11, pp. 31 -- 35, 2000.

\bibitem{TAYLOR1991}
H.~Taylor, ``Florentine rows or left-right shifted permutation matrices with
  cross-correlation values $\leq$ 1,'' \emph{Discrete Mathematics}, vol.~93,
  no.~2, pp. 247--260, 1991.

\bibitem{Colbourn}
C.~J. Colbourn and J.~H. Dinitz, \emph{Handbook of Combinatorial Designs,
  Second Edition (Discrete Mathematics and Its Applications)}.\hskip 1em plus
  0.5em minus 0.4em\relax Chapman Hall/CRC, 2006.

\bibitem{HYSong}
H.~Y. {Song}, ``On aspects of tuscan squres,'' \emph{PhD thesis, University of
  Southern California}, 1991.

\bibitem{GolombT}
S.~Golomb, T.~Etzion, and H.~Taylor, ``Polygonal path constructions for
  tuscan-k squares,'' \emph{Ars Combinatoria}, vol.~30, pp. 97--140, 1990.

\bibitem{Kumar}
P.~V. Kumar, R.~A. Scholtz, and L.~R. Welch, ``Generalized bent functions and
  their properties,'' \emph{Journal of Combinatorial Theory, Series A},
  vol.~40, no.~1, pp. 90--107, 1985.

\bibitem{chirp-like}
B.~M. Popovic and O.~Mauritz, ``Generalized chirp-like sequences with zero
  correlation zone,'' \emph{IEEE Transactions on Information Theory}, vol.~56,
  no.~6, pp. 2957--2960, 2010.

\bibitem{Milewski}
A.~{Milewski}, ``Periodic sequences with optimal properties for channel
  estimation and fast start-up equalization,'' \emph{IBM Journal of Research
  and Development}, vol.~27, no.~5, pp. 426--431, Sep. 1983.

\bibitem{Chung}
H.~{Chung} and P.~V. {Kumar}, ``A new general construction for generalized bent
  functions,'' \emph{IEEE Transactions on Information Theory}, vol.~35, no.~1,
  pp. 206--209, Jan 1989.

\bibitem{Frank}
R.~L. Frank and S.~Zadoff, ``Phase shift pulse codes with good periodic
  correlation properties,'' \emph{IRE Trans. Inform. Theory}, vol.~8, no.~6,
  pp. 381--382, 1962.

\bibitem{Heimiller}
R.~Heimiller, ``Phase shift pulse codes with good periodic correlation
  properties,'' \emph{IRE Transactions on Information Theory}, vol.~7, no.~4,
  pp. 254--257, 1961.

\end{thebibliography}
\begin{IEEEbiographynophoto}{Dan Zhang}
received the B.S. and M.S. degrees in mathematics from Henan University,
Kaifeng, China, in 2011 and 2014 respectively.
From Sept. 2012 to June 2014, she was a visiting student in the Academy of Mathematics and System Science, Chinese Academy of Sciences, China. She is currently a Ph.D. student in the Department of Informatics at the University of Bergen, Norway.
Her research interests lie in sequence design, cryptography, and coding theory.
\end{IEEEbiographynophoto}

\end{document}